\documentclass[twocolumn,amsmath]{revtex4-1}
\usepackage{amsthm,amsfonts,hyperref,cleveref,braket,graphicx, bm}
\newtheorem{theorem}{Theorem}[section]
\newtheorem{lemma}[theorem]{Lemma}
\theoremstyle{definition}
\newtheorem{definition}[theorem]{Definition}
\newtheorem{assumption}[theorem]{Assumption}
\newtheorem{remark}[theorem]{Remark}
\newcommand{\Hilb}[1][]{\mathcal{H}_{#1}}
\newcommand{\chn}{\mathcal{E}}
\newcommand{\Ps}{\mathcal{P}}
\newcommand{\Ts}{\mathcal{T}}
\newcommand{\Ms}{\mathcal{M}}
\newcommand{\Pns}{\mathcal{P}_\text{NS}}
\DeclareMathOperator{\tr}{Tr}
\newcommand{\Tr}[2][]{\tr_{#1} \left ( #2 \right )}
\newcommand{\proj}[1]{\ket{#1}\bra{#1}}
\begin{document}

\title{Is a time symmetric interpretation of quantum theory possible without retrocausality?}
\author{Matthew S. Leifer}
\affiliation{Institute for Quantum Studies \& Schmid College of Science and Technology, Chapman University, Orange, CA 92866, USA}
\author{Matthew F. Pusey}
\affiliation{Perimeter Institute for Theoretical Physics, 31 Caroline Street North, Waterloo, ON N2L 2Y5, Canada}
\date{May 10, 2017}
\begin{abstract}
Huw Price has proposed an argument that suggests a time-symmetric ontology for quantum theory must necessarily be retrocausal, i.e.\ it must involve influences that travel backwards in time. One of Price's assumptions is that the quantum state is a state of reality.  However, one of the reasons for exploring retrocausality is that it offers the potential for evading the consequences of no-go theorems, including recent proofs of the reality of the quantum state.  Here, we show that this assumption can be replaced by a different assumption, called $\lambda$-mediation, that plausibly holds independently of the status of the quantum state.  We also reformulate the other assumptions behind the argument to place them in a more general framework and pin down the notion of time symmetry involved more precisely.  We show that our assumptions imply a timelike analogue of Bell's local causality criterion and, in doing so, give a new interpretation of timelike violations of Bell inequalities.  Namely, they show the impossibility of a (non-retrocausal) time-symmetric ontology.
\end{abstract}
\maketitle

\section{Introduction}

In \cite{Price2012}, Huw Price gave an argument to the effect that any time symmetric ontology for quantum theory must involve retrocausality (influences that travel backwards in time).  Price's argument is based on analysing an experiment where a single photon is passed through two polarizing beamsplitters in sequence.

One of Price's assumptions, which he calls \emph{Realism} \footnote{We use italics for Price's assumptions and bold for ours, as we later define a weaker notion of \textbf{Realism}.}, is really an assumption of the reality of the quantum state.  Briefly, he assumes that the usual forwards-evolving polarization vector of a photon emerging from a polarizing beam splitter is a real physical property of the photon.  This is an assumption of the reality of the quantum state because the polarization vector is isomorphic to the part of the quantum state that describes the polarization degree of freedom. Price uses this, together with his other assumptions, to argue that there must also be a backwards evolving polarization vector coming from the second beam splitter that is oriented along its measurement direction, and hence there is retrocausality.   

The reality of both forwards and backwards evolving states implied by Price's argument occurs in the transactional interpretation of quantum theory \cite{Cramer1986, Kastner2012}, and is reminiscent of the two state vector formalism of Aharonov and collaborators \cite{Aharonov1964, Aharonov2008}.  Indeed, if we interpret the two state vectors in the formalism as a description of reality, rather than just a convenient mathematical way of re-expressing the predictions of quantum theory in situations of pre- and post-selection, then this theory satisfies Price's assumptions.

However, one possible response to Price's argument is to simply deny the reality of the quantum state, in favor of the $\psi$-epistemic view, which says that the quantum state is not a physical property of a system, but rather something more akin to a probability distribution, such as a Liouville distribution in statistical mechanics.  In other words, it describes our knowledge about the quantum system rather than being an intrinsic property of an individual system.  If we do not believe that the forwards evolving state vector is a state of reality, than neither Price's aregument nor the two state vector formalism, provides evidence that we should view the backwards evolving state vector as a state of reality. The $\psi$-epistemic view provides natural explanations for many quantum phenomena that are otherwise puzzling, such as the indistinguishability of quantum states and the no-cloning theorem, which have been discussed in detail elsewhere \cite{Spekkens2007, Leifer2014}.  Although several recent results \cite{Pusey2012, Aaronson2013, Barrett2013, Branciard2014, Colbeck2012, Colbeck2013a, Hardy2013, Leifer2014a, Leifer2013c, Mansfield2014, Maroney2012, Maroney2012a, Montina2015} (see \cite{Leifer2014} for a review) cast doubt on the viability of the $\psi$-epistemic view, all of these results assume that there is no retrocausality, so allowing for retrocausality is a potentially appealing way of saving the $\psi$-epistemic view.  Because of this, it would be preferable not to rule out $\psi$-epistemic theories a priori when making an argument for retrocausality.

The first aim of this paper, is to construct an alternative to Price's argument, which does not assume the reality of the quantum state.  This turns out to be possible using a different assumption, that we call $\bm{\lambda}$\textbf{-mediation}, which plausibly holds independently of the status of the quantum state.  Briefly, it says that any correlations between a preparation and a measurement made on a system should be mediated by the physical properties of the system.  There is a close parallel between the relationship between Price's argument and ours, and the relationship between the Einstein-Podolsky-Rosen (EPR) argument and Bell's theorem.  In fact, our assumptions imply a timelike analogue of Bell's local causality condition.  It is well-known that sequences of measurements made on the same system can violate Bell inequalities, the significance of which has long been debated \cite{Leggett1985, Brukner2004, Taylor2004, Lapiedra2006, Spekkens2009, Avis2010, Fritz2010, Budroni2013, Zukowski2013, Budroni2014, Markiewicz2014, Maroney2014, Brierley2015}.  Our argument provides a new spin on this --- temporal Bell violations can be viewed as proofs of the impossibility of a time-symmetric ontology (unless we allow retrocausality).

In fact, it is possible to prove a weaker result even without $\bm{\lambda}$\textbf{-mediation}, based on the Colbeck-Renner theorem \cite{Colbeck2011, Colbeck2012a}, which we discuss in \cref{app}.

The second aim of this paper is to clarify and extend the scope of the other assumptions behind Price's argument.  Like Bell's theorem, we are aiming for a result that is independent of the details of quantum theory, so that we can say whether or not a given operational theory allows for a time symmetric ontology.  In order to do this, we have to replace one of Price's assumptions, which he calls \emph{Discreteness} by a broader principle that can be formulated at this level of generality.  Doing so clarifies the precise notion of time-symmetry that is at stake in the argument.

In particular, an obvious objection to both Price's and our argument is that unitary evolution according to the Schr{\"o}dinger equation is time symmetric, in the conventional sense of time-symmetry of dynamical laws.  Therefore, an interpretation like Everett/many-worlds \cite{Everett1957, Wallace2012, Vaidman2016}, in which a unitarily evolving quantum state is the entire ontology, is time-symmetric without requiring retrocausality.  The same is true of interpretations that supplement this picture with additional variables that also evolve under time-symmetric equations of motion, such as de Broglie-Bohm theory \cite{Broglie2009, Bohm1952, Bohm1952a, Duerr2009}.  Hence, a direct inference from the conventional notion of time symmetry to retrocausality is out of the question.  

In Price's version of the argument, these theories are ruled out by \emph{Discreteness}.  This says that if a photon is detected at one of the output ports of a polarizing beamsplitter then nothing comes out of the other ports of the beamsplitter.  This assumption is needed so that a photon exiting a beamsplitter and then being detected can be treated as the time reverse of inserting a photon into a beamsplitter through a definite port.  Everett and de Broglie-Bohm violate this assumption because, when a photon exits a beamsplitter, there are branches of the quantum state that exit along each of the ports.  These represent what will become different worlds in Everett, and they are empty waves that do not contain a particle in de Broglie-Bohm.  

We also note that there is some tension between the \emph{Discreteness} and \emph{Realism} assumptions.  \emph{Realism} requires that the quantum state of the system should be real while it is travelling between the two beamsplitters.  If this is true then it would be natural to require that the quantum state exiting the second beamsplitter should also be real, even though this is not strictly required by \emph{Realism}.  However, this would contradict \emph{Discreteness}.  The tension is removed if we allow for $\psi$-epistemic models, since then the branch of the quantum state that exits the port of the beamsplitter on which the photon is not detected need not correspond to anything real \footnote{In any case, the relevant time symmetry does not require \emph{Discreteness}.  Briefly, if we allow that the vacuum may be represented by a probability distribution over more than one ontic state then there is something nontrivial that enters both ports of the first beamsplitter when we insert the photon along a definite port.  Then, having something exit both ports of the second beamsplitter need not prevent the necessary time symmetry.  See \cite{Martin2016} for an explanation of quantum interference along these lines.}.

In our view, the symmetry between preparations and measurements in quantum theory is the main motivation for \emph{Discreteness}, so we prefer to base our version of the argument on the former, as it makes it clearer that the argument involves a stronger notion of time symmetry than the conventional one.  \emph{Discreteness} itself would be difficult to formulate at the level of generality we are trying to achieve in any case.  It is known that sequences of quantum experiments can be described by a retrodictive formalism \cite{Barnett2000, Pegg2002, Pegg2002a, Leifer2013a}, in which quantum states are associated with measurement outcomes and are evolved backwards in time to preparations, which are in turn described by what we normally think of as ``measurement'' operators.  The symmetry in question is that the retrodictive formalism is mathematically identical to the conventional predictive formalism, in which quantum states are evolved forwards in time from preparation to measurement, so every quantum experiment has a kind of time reverse, in which the roles of the predictive and retrodictive operators are exchanged.  A simple example of follows from the observation that $|\braket{\phi|\psi}|^2 = |\braket{\psi|\phi}|^2$.  Consider two experiments, the first of which consists of preparing one state uniformly at random from a complete orthonormal basis $\{\ket{\psi_n}\}$, and then measuring in another orthonormal basis $\{\ket{\psi_m}\}$.  In the second experiment, the order is reversed, so that a state is prepared uniformly at random from $\{\ket{\phi_m}\}$ and then measured in the basis $\{\ket{\psi_n}\}$.  Because $|\braket{\phi_m|\psi_n}|^2=|\braket{\psi_n|\phi_m}|^2$, the joint probability of obtaining $\ket{\psi_n}$ and $\ket{\phi_m}$ is the same in both experiments.  Hence, the predictive probability of obtaining $\ket{\phi_m}$ given that $\ket{\psi_n}$ was prepared in the first experiment is the same as the retrodictive probability of having prepared $\ket{\phi_m}$ given that the measurement outcome was $\ket{\psi_n}$ in the second experiment. Our \textbf{Time Symmetry} assumption says that this symmetry, which is a symmetry of the operational predictions of quantum theory, should be taken seriously as a fundamental symmetry, and therefore imposed at the ontological level as well.  We formulate this assumption rigorously, independently of the details of quantum theory, and base our argument on it.  Note that \textbf{Time Symmetry} is different from the usual notion of time symmetry of dynamical laws, and neither Everett or de Broglie-Bohm satisfy the former.

Given this, one might question whether \textbf{Time Symmetry} is really necessary.  However, our assumption does bear a family resemblance to the assumptions behind other no-go theorems that are often regarded as reasonable.  In particular, Spekkens' notion of noncontextuality says that experimental procedures that are operationally indistinguishable ought to be represented in the same way at the ontological level \cite{Spekkens2005}.  \textbf{Time Symmetry} could be motivated by a principle that extends this to say that a symmetry of the operational predictions ought to also hold at the ontological level.  Spekkens' principle is then just the special case where the symmetry is identity of operational predictions.  Whilst one might quibble about applying this principle to every possible symmetry of the operational predictions, some of which may be accidental, it seems reasonable to apply it to fundamental symmetries, and time symmetry is pretty fundamental.  We therefore think that it is surprising that a conventional realist (non retrocausal) interpretation cannot satisfy \textbf{Time Symmetry}, and that the fact that Everett and de Broglie-Bohm do not is a genuine deficiency of those theories.

The remainder of this paper is structured as follows.  In \S\ref{TS}, we outline several notions of time symmetry and explain how the notion used in this paper differs from the conventional one.  \S\ref{OT} introduces our operational formalism for prepare-and-measure experiments.  \S\ref{OpTS} formally defines our operational notion of time symmetry and shows that quantum theory satisfies it.  \S\ref{OM} and \S\ref{OnTS} mirror this structure at the ontological level.  \S\ref{OM} introduces our main assumptions about realist models of operational theories, including \textbf{No Retrocausality} and \textbf{$\bm{\lambda}$-mediation}, and \S\ref{OnTS} defines our notion of ontological time symmetry.  \S\ref{TSA} explains our \textbf{Time Symmetry} assumption, which is that operational time symmetries should be reflected at the ontological level.  \S\ref{MR} contains our main results.  \S\ref{MT} proves that models satisfying our assumptions must obey Bell's local causality condition for timelike experiments.  \S\ref{QE} gives a simple example of how quantum theory violates this, and \S\ref{QG} shows that there is an isomorphism between spacelike and timelike experiments, which allows us to port all existing results about violation of bipartite Bell inequalities with entangled states to the timelike context.  In \S\ref{Price}, we compare our assumptions with Price's, and reformulate the logic of his argument in our framework.  \S\ref{EdBB} explains why de Broglie-Bohm and Everett do not satisfy \textbf{Time Symmetry}, but this does not conflict with their dynamical time symmetry.  \S\ref{Spekkens} discusses the relation of our main result with Spekkens' notion of preparation contextuality, and discusses the relevance of experiments that have recently been conducted to test this.  In \S\ref{Conc}, we conclude by discussing which assumption ought to be given up in the light of our results.  \Cref{app} discusses how the Colbeck-Renner theorem can be used to obtain a weaker result without \textbf{$\bm{\lambda}$-mediation}.

\section{Time symmetry}

\label{TS}

The usual notion of time symmetry used in physics is a symmetry of the dynamics, wherein there is a one-to-one correspondence between solutions to the equations of motion and solutions to the modified equations obtained by replacing $t$ with $-t$.  However, this is not the only notion of time symmetry worth considering.  We start from the more informal idea that time symmetry is the inability to tell whether a video of a physical process is running forwards or backwards, and note that this can be formalized in several different ways.

First of all, it matters if one really means an actual video, which can only capture things that are observable, or whether the video is merely a conceptual stand in for a complete record of everything that exists.  We call the former kind of time symmetry an \emph{operational} time symmetry and the latter an \emph{ontological} time symmetry.  Our argument makes use of both kinds of time symmetry, and is based on the assumption that one should take an operational symmetry as evidence for an ontological one.  We formalize operational time symmetry in \S\ref{OpTS}, after introducing our operational formalism in \S\ref{OT}, and we formalize ontological time symmetry in \S\ref{OnTS}, after we have introduced the ontological framework in \S\ref{OM}.

Secondly, it makes a difference whether one is only interested in whether a given video is \emph{possible} in both time directions, or whether there is a process that makes it equally \emph{likely} in both time directions.  A smashed glass spontaneously reconstituting itself is not impossible according to the laws of physics.  According to the usual notion of dynamical time symmetry, a spontaneously reconstituting glass is a possible solution to the equations of motion whenever an unbroken glass smashing into pieces is.  However, a spontaneously reconstituting glass is vastly less likely, given the low entropy initial conditions of the universe.  In light of this, we can consider a stronger notion of time symmetry in which we also demand that the probability of occurrence should be the same in both time directions.  This holds, for example, for systems that are in thermal equilibrium.  It is this stronger notion of time-symmetry that we will be considering here.  Of course, the universe as a whole does not display such time-symmetry, and we will not be assuming that it does.  However, it is possible to construct experiments that have this time-symmetry at the operational level.  Our main assumption will be that, if this is the case, the same symmetry ought to hold at the ontological level as well.

\section{Operational formalism}

\label{OT}

We will be discussing experiments consisting of three stages, a preparation $P$, a transformation $T$, and a measurement $M$ (see \cref{opfig}). Each preparation and measurement has an input, which is under the control of the experimenter, and an output, which they do not directly control, i.e.\ the outputs may be correlated with the inputs, but the experimenter does not have any further control over the outputs other than via choosing the inputs.  As an example, in Price's argument, both the preparation and measurement devices are polarizing beam splitters and, in both cases, the inputs are the choice of angle at which to orient the beamsplitter.  For the measurement, the output is which output port of the beamsplitter the photon is detected on.  For the preparation, the output is the choice of which \emph{input} port to insert the photon into the beamsplitter from.  Note here that our distinction between input and output should be thought of as the distinction between controlled and uncontrolled variables, and does not necessarily represent the order in which things happen to the photon.  The photon is inserted into the input port before it encounters the beamsplitter oriented at a certain angle.  Nonetheless, the angle is our ``input'' and the port is our ``output'' because the former is controlled by the experimenter and the latter is not.  The choice of input port is in principle controllable by the experimenter but, in the specific setup used by Price, it is not actually controlled by them, but rather by a third party, which Price calls the ``demon of the left''.  For this reason, it is important not to think of the input as consisting of everything that is in principle controllable by the experimenter and the outputs as everything that is in principle uncontrollable, but rather the inputs are the variables that they actually control and the outputs are variables that they do not actually control, even if some of them are in principle controllable.  To make the argument work, it is necessary to place constraints on what the experimenter controls.  We shall return to the reason and justification for this in \S\ref{OpTS}.

Although, in Price's argument, the output of the preparation represents something that happens to the photon before the input, it is important to note that the input is always chosen before the output.  The experimenter decides how to orient the beamsplitter and then the ``demon of the left'' chooses on which port to insert the photon.  In general, we should think of preparations and measurements as black-box devices, which have a setting switch that allows the experimenter to choose the input and a pointer or screen that subsequently displays the output.  Regardless of what is happening inside the black box, the operational act of choosing the input always happens before the output is displayed, so, in a non retrocausal theory, the input can causally influence the output but not vice versa.

\begin{figure}[!htb]
\includegraphics[scale=0.8]{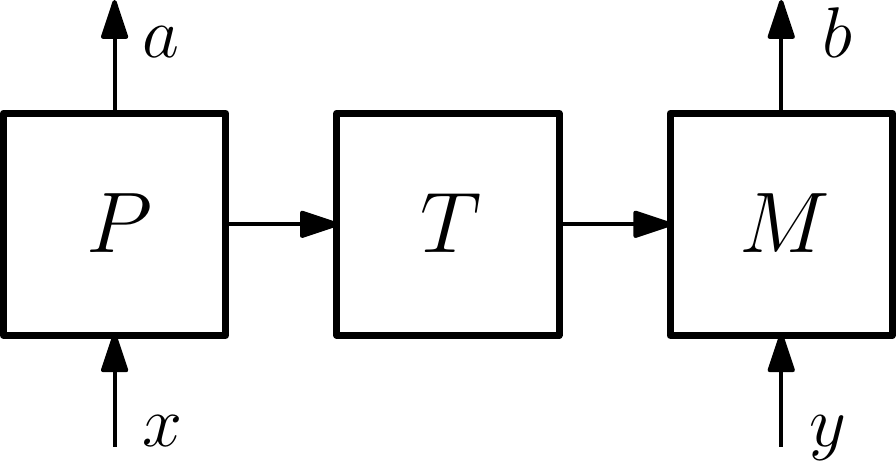}
\caption{\label{opfig}Illustration of prepare-transform-measure experiment.}
\end{figure}

The final piece of the setup is the transformation, which just represents the dynamics of the system between preparation and measurement.  In Price's example, this is trivial as the polarization degree of freedom does not evolve as it passes between the two beamsplitters, but we will consider what happens when dynamics, which may include decoherence and losses to the environment, is added to the picture.

An \emph{operational theory} specifies a set $\Ps$ of possible preparations, a set $\Ms$ of possible measurements, and a set $\Ts$ of possible transformations. In general, a theory may include different types of systems, e.g.\ in quantum theory we have photons, atoms, etc.\ that have different degrees of freedom and are described by Hilbert spaces of differing dimensionality.  Therefore, not every preparation, transformation, and measurement need be compatible with one another, e.g.\ it does not make sense to use a device for measuring the energy of an atom to measure the polarization of a photon.  We call a triple $(P,M,T)$ of preparation, measurement, and transformation that are compatible with each other an \emph{experiment}.  For every experiment $(P,M,T)$, the operational theory must also specify a prediction of the probabilities of the output variables of $P$ and $M$, given the choice of input variables.  We write these probabilities as $p_{PMT}(a,b|x,y)$, where $x$ is the input of $P$ and $a$ is its output, and $y$ is the input of $M$ and $b$ is its output.

For example, in quantum theory, a preparation $P$ is associated with a Hilbert space $\Hilb[A]$, a set of (unnormalized) density operators $\{\rho_{aA|x}\}$ on $\Hilb[A]$ --- one for each choice of $x$ and $a$ --- such that the ensemble average density operators $\rho_{A|x} = \sum_a \rho_{aA|x}$ are normalized $\Tr{\rho_{A|x}} = 1$.  The preparation procedure starts with the experimenter choosing $x$.  The preparation device then generates a classical variable $a$ with probability distribution $p(a|x) = \Tr{\rho_{aA|x}}$, outputs $a$, and prepares the system in the corresponding (normalized) state $\rho_{aA|x}/p(a|x)$, which is subsequently fed into the transformation device.

A measurement $M$ is described by a set $\{E_{b|yB}\}$ of Positive Operator Valued Measures (POVMs) on a Hilbert space $\Hilb[B]$ --- one POVM for each choice of $y$.  This means that for each $(y,b)$, $E_{b|yB}$ is a positive operator, and for all $y$, $\sum_b E_{b|yB} = I_B$, where $I_B$ is the identity operator on $\Hilb[B]$.

A transformation is described by a completely-positive, trace-preserving (CPT) map $\mathcal{E}$ from the density operators on the input Hilbert space $\Hilb[\text{in}]$ to the density operators on the output Hilbert space $\Hilb[\text{out}]$.  We allow $\Hilb[\text{in}]$ and $\Hilb[\text{out}]$ to differ because the transformation may involve discarding part of the system into the environment or adding ancillas.

A triple $(P,M,T)$ is an experiment only if $\Hilb[\text{in}] = \Hilb[A]$ and $\Hilb[\text{out}] = \Hilb[B]$, in which case we write the CP-map as $\mathcal{E}_{B|A}$.

For such an experiment, by the Born rule, quantum theory predicts the probabilities
\begin{equation}
    p_{PMT}(a,b|x,y) = \Tr{E_{b|yB} \mathcal{E}_{B|A} \left ( \rho_{aA|x}\right )}.
\end{equation}

\section{Operational time-symmetry}

\label{OpTS}

\begin{definition}
  \label{def:opTR}
  An experiment $(P,M,T)$ has an \emph{operational time reverse} if there exists another experiment $(P',M',T')$ where $P'$ has the same set of inputs and outputs as $M$, $M'$ has the same set of inputs and outputs as $P$, and
  \begin{equation}
    p_{P'M'T'}(b,a|y,x) = p_{PMT}(a,b|x,y)
  \end{equation}

  An operational theory has \emph{operational time symmetry} if every experiment has an operational time reverse.
\end{definition}

For each run of an experiment, consider a record of the inputs and outputs of the preparation and measurement $(x,a,y,b)$.  The experiment is repeated an arbitrarily large number of times with each choice of inputs and the records are presented to a third party, without telling them what the preparation, measurement, and transformation actually is.  We also give them the same records with the order of preparation and measurement variables reversed $(y,b,x,a)$.  If, knowing the theory, they cannot tell which is the true record and which is the reversed one, then the theory has operational time symmetry.

One might think that inputs and outputs also ought to be swapped in the reversed record.  However, recall that our notion of input and output does not directly relate to time ordering, but is rather to do with controlled vs.\ uncontrolled variables.  For a preparation, all variables are in principle controllable, so the specification of which variables are inputs and which are outputs is fairly arbitrary at this stage.  However, we shall be particularly interested in experiments in which the experimenter has the same degree of control over a measurement as they do over a preparation, in which case the distinction between controlled and uncontrolled variables will be the same in both the true record and the time reversed one.

Note that, we are not assuming that every experiment \emph{ought} to have an operational time reverse.  At this point, it is just a definition.  In fact, it would be unnatural to impose this for the simple reason that it is possible to send a signal forwards in time, but not backwards in time.  Formally, we expect an operational theory to satisfy
\begin{equation}
    \label{eq:nosigforwards}
    \sum_b p_{PMT}(a,b|x,y) = \sum_b p_{PMT}(a,b|x,y'),
\end{equation}
for all $y$ and $y'$, so that the output of the preparation cannot be used to infer any information about the input of the measurement, but there is no reason why the output of the measurement should not contain information about the input of the preparation, so we do not expect
\begin{equation}
    \sum_a p_{PMT}(a,b|x,y) = \sum_a p_{PMT}(a,b|x',y),
\end{equation}
to hold in general for $x \neq x'$.  If it does not hold, then \cref{eq:nosigforwards} prevents the experiment from having an operational time reverse.

However, we do not believe that the lack of operational time symmetry is a fundamental asymmetry of physics, but is rather a consequence of the thermodynamic arrow of time.  Specifically, an operational description of an experiment is, by definition, a description of the experiment from the point of view of the experimenter.  The experimenter has a subjective arrow of time.  This means that she can remember the past but not the future, and also that her past seems fixed but her future seems open and unpredictable.  It is usually thought that the subjective arrow of time is a consequence of the thermodynamic arrow because both memory and control consume a supply of external low entropy systems.  For example, turning a knob from an initially unknown position to a definite setting is an erasure process, which increases entropy by Landauer's principle \cite{Landauer1961}.

Because of this, an experimenter typically has more control over the preparation than the measurement, but there is reason to believe that this is an emergent, rather than fundamental, asymmetry.  Hypothetically, if it were possible for the experimenter to have a subjective arrow of time pointing in the opposite direction to the arrow of time used to describe the rest of the experiment, then presumably they would be able to control the measurement outcome to the same degree that we can control the preparation.

In order to identify the fundamental operational time symmetries, we want to factor out this emergent asymmetry in some way.  One way of doing this would be to simulate full control over the measurement outcome by allowing post-selection.  Although we could do things this way, post-selection significantly complicates the analysis.  For our purposes, it is sufficient to instead restrict the experimenter's control over the preparation procedure, so that she cannot send signals forwards in time.  Doing this plays the same role as the ``demon of the left'' in Price's argument.

\begin{definition}
    A preparation $P$ is \emph{no-signalling} if, for all $M \in \Ms$ and $T \in \Ts$ such that $(P,M,T)$ is an experiment,
    \begin{equation}
        \sum_a p_{PMT}(a,b|x,y) = \sum_a p_{PMT}(a,b|x',y),
    \end{equation}
    for all $x$ and $x'$.

    $\Pns$ denotes the set of all no-signalling preparations in an operational theory.  The \emph{no-signalling sector} of an operational theory is the operational theory in which $\Ps$ is restricted to $\Pns$
\end{definition}

Note that the no-signalling sector of an operational theory need not be operationally time symmetric.  However, it is in both quantum theory and classical probability theory.  This is therefore a nontrivial time symmetry that our present physical laws obey, which is different from the conventional notion of time-symmetric dynamical laws.

\begin{theorem}
    \label{thm:QTS}
    The no-signalling sector of quantum theory is operationally time-symmetric.
\end{theorem}

\begin{proof}
    First, note that a quantum preparation is no-signalling iff the ensemble average state is the same for all inputs, otherwise there would be a measurement that predicts different probabilities for different preparation inputs.  Therefore, we can define $\rho_A = \rho_{A|x}$, which is independent of $x$.

    Given an experiment $(P,M,T)$, where $P$ is no-signalling, we have to define an operational time reverse $(P',M',T')$ that predicts the same probabilities.  To do this, we make use of quantum versions of Bayes' rule, as described in \cite{Leifer2013a}.

    Specifically, let $\rho_{aA|x}$ be the states associated with $P$, let $E_{b|yB}$ be the POVM elements associated with $M$, and let $\mathcal{E}_{B|A}$ be the CPT map associated with $T$.  We first define the POVM elements $E_{a|xA}$ associated with $M'$ via
    \begin{equation}
        \label{eq:retroPOVM}
       E_{a|xA} = \rho_A^{-\frac{1}{2}} \rho_{aA|x} \rho_A^{-\frac{1}{2}}.
    \end{equation}
    This should be understood as an equation on the subspace $\text{supp}(\rho_A)$ of $\Hilb[A]$, so that $\rho_A$ has a well defined inverse.  It is easy to check that these operators are positive and $\sum_a E_{a|xA} = I_{\text{supp}(\rho_A)}$.  If we want $\Hilb[A]$ to be the Hilbert space associated with $M'$ then the operators $E_{a|xA}$ may be extended by adding an arbitrary POVM on the orthogonal subspace to them.  Alternatively, we can just let $\text{supp}(\rho_A)$ be the Hilbert space of $M'$.

    Next, it is convenient to define the ensemble average state at the output of the transformation as $\rho_B = \mathcal{E}_{B|A}(\rho_A)$.  Then we can define the states associated with $P'$ as
    \begin{equation}
    \label{eq:retrostates}
    \rho_{bB|y} = \rho_B^{\frac{1}{2}} E_{b|yB} \rho_B^{\frac{1}{2}}.
    \end{equation}
    It is easy to check that these operators are positive and $\sum_b \rho_{bB|y} = \rho_B$, so the ensemble average is a normalized density operator.

    Finally, we define the CPT map associated with $T'$ as
    \begin{equation}
    \chn_{A|B}(\, \cdot \,) = \rho_A^{\frac{1}{2}} \mathcal{E}_{B|A}^{\dagger} \left ( \rho_B^{-\frac{1}{2}} \left ( \, \cdot \, \right ) \rho_B^{-\frac{1}{2}} \right ) \rho_A^{\frac{1}{2}},
    \end{equation}
    where $\chn^{\dagger}_{B|A}$ is the adjoint of $\chn_{B|A}$, i.e.\ the unique map that satisfies $\Tr{M_A\chn_{B|A}^{\dagger}(N_B)} = \Tr{N_B \chn_{B|A}(M_A)}$ for all operators $M_A$ on $\Hilb[A]$ and all operators $N_B$ on $\Hilb[B]$.  The map $\mathcal{E}_{A|B}$ should be understood as a map from density operators on $\text{supp}(\rho_B)$ to density operators on $\Hilb[A]$, and we can extend it to $\Hilb[B]$ arbitrarily if necessary.  The Choi-Jamio{\l}kowski isomorphism \cite{Jamiolkowski1972, Choi1975} can be used to verify that this is a well-defined CPT map, as was done in \cite{Leifer2013a}.

    With these definitions, it is straightforward to verify that
    \begin{multline}
        p_{P'M'T'}(b,a|y,x) = \Tr{E_{a|xA}\chn_{A|B}(\rho_{bB|y})} \\ = \Tr{E_{b|yB}\chn_{B|A}(\rho_{aA|x})} = p_{PMT}(a,b|x,y),
    \end{multline}
    as required.
\end{proof}

\section{Ontological Formalism}

\label{OM}

So far, we have only been discussing the operational predictions of a theory.  We now wish to consider the possibility that the system has some real physical properties between preparation and measurement, that are responsible for mediating any correlations between the two.

There is a standard framework for discussing this, known as the \emph{ontological models framework} \cite{Harrigan2010}.  We think that this is a well motivated framework, but because the ontological models framework implicitly assumes that there is no retrocausality, and the necessity of retrocausality is precisely the point at issue, we will break things down into four more primitive assumptions, so that the role of the no retrocausality assumption can be seen more clearly.  This will also facilitate a comparison with Price's assumptions in \S\ref{Price}.

We will use influence diagrams (a generalization of Bayesian networks that allows for decision nodes) \cite{Howard2005} to illustrate our assumptions, but familiarity with such diagrams is not assumed.

\begin{assumption}
    (Single-world) \textbf{Realism}: The system has some physical properties, a specification of which is called its \emph{ontic state}, denoted $\lambda$.  Ontic states take values in a (measurable) set $\Lambda$ called the \emph{ontic state space} \footnote{We adopt ``physicists' rigour'' and ignore measure-theoretic complications here.  They can be handled along the lines of the treatment in \cite{Leifer2014}.}.

    On each run of an experiment, the operational variables (inputs $x$ and $y$, and outputs $a$ and $b$) and the ontic state $\lambda$ each take a definite value.  Our uncertainty about these variables can be described according to classical probability theory.
\end{assumption}

The idea here is that $\lambda$ represents properties of the system that exist, but may not be directly observable.  For example, in a classical experiment in statistical mechanics, we may only have operational access to variables describing the macrostate; such as pressure, volume, and temperature; and the ontic state would be the microstate, i.e.\ a specification of the positions and momenta of each particle.  The idea is that perhaps the ways in which we can manipulate a quantum system are limited in a similar way.

The exact ontic state may be unknown to the experimenter, but it is in principle knowable, even if we cannot actually construct an experiment that would reveal it.  It is for this reason that we consider the use of classical probability to be part and parcel of the definition of realism.  Because $\lambda$ takes a specific value in each run of experiment, you can in principle resolve a bet on its value and it has a definite relative frequency in a sequence of runs, so however you prefer to interpret probability theory, the standard arguments for its use apply.  Although it is sometimes posited that quantum theory should be understood in terms of a non-standard probability theory, we have yet to see an interpretation of such a theory that is realist in character.

Our realism is single world because we assume that the full specification of the ontology includes definite values for the operational variables, e.g.\ measurement outcomes.  This rules out Everett/many-worlds.  However, adopting Everett/many-worlds is not actually a way out of our no-go theorem, because, as explained in the introduction, it does not have the kind of time symmetry we are interested in.  Therefore, it is possible that the argument presented here is generalizable to interpretations that involve multiple worlds, but we do not know how to do this yet.

\begin{assumption}
    \textbf{Free Choice}:  The experiment can be described by a conditional probability distribution $p_{PMT}(a,b,\lambda|x,y)$.  We assume that the experimenter is free to choose $x$ and $y$ however they like.
\end{assumption}

What we mean by free choice of $x$ and $y$ is that the experimenter can set the probabilities $p(x,y)$ for the inputs in any way that she likes, so that the joint probabilities are $p_{PMT}(a,b,\lambda,x,y) = p_{PMT}(a,b,\lambda|x,y)p(x,y)$.  All possible combinations of inputs can be tested in as many runs as she likes \footnote{For this, it is sufficient to choose $x$ and $y$ independently and uniformly at random, as is done in an ideal Bell experiment.}.  In the language of influence diagrams, $x$ and $y$ are decision nodes, represented by squares in \cref{fig:noretro}.  In particular, $x$ and $y$ have no direct causes, i.e.\ no parents in the graph.

\begin{definition}
    Given an experiment $(P,M,T)$, a specification of an ontic state space $\Lambda$ together with a conditional probability distribution $p_{PMT}(a,b,\lambda|x,y)$ is called an \emph{ontic extension} of the experiment.  The ontic extension is required to reproduce the predictions of the operational theory upon marginalizing over $\lambda$, i.e.\
    \begin{equation}
        \int_{\Lambda} p_{PMT}(a,b,\lambda|x,y)d\lambda = p_{PMT}(a,b|x,y).
    \end{equation}

    We denote an ontic extension by $(P,M,T,\Lambda)$, leaving the probabilities implicit.

    A choice of ontic extension for each experiment in an operational theory is called an ontic extension of the theory.
\end{definition}

\textbf{Realism} and \textbf{Free Choice} are assumed in the vast majority of work on realist approaches to quantum theory.  We shall therefore tend to just assume that experiments have ontic extensions, rather than mentioning these two assumptions explicitly.

For the remainder of this section, we will concentrate on a single experiment $(P,M,T)$, so we will drop the subscripts on the probability distributions.

\begin{assumption}
    \textbf{No Retrocausality}: All non-input variables should be conditionally independent of input variables in their future, given a specification of all the variables in their past.
\end{assumption}

The possible causal dependencies that are compatible with no-retrocausality are illustrated in \cref{fig:noretro}.  Our assumption can be understood as saying that $p(a,b,\lambda|x,y)$ should obey the causal Markov condition with respect to this diagram.

\begin{figure}[!htb]
    \centering
    \includegraphics[width=0.9\columnwidth]{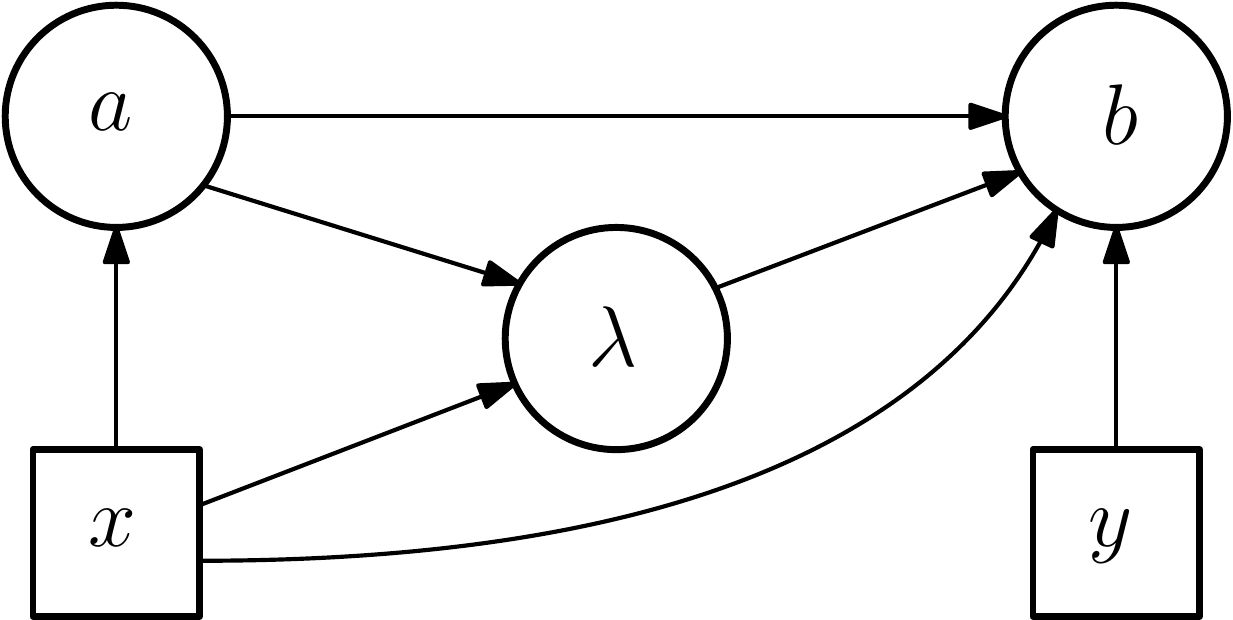}
    \caption{\label{fig:noretro}An influence diagram representing the possible causal influences in a model with no retrocausality.  A square represents a variable that is under the direct control of the experimenter and a circle represents a variable that they do not control.  An arrow between two nodes $u$ and $v$ in the diagram represents the possibility that $u$ can be a direct cause of $v$.}
\end{figure}

In more detail, a general conditional probability distribution $p(a,b,\lambda|x,y)$ can be decomposed as
\begin{equation}
    \label{eq:gendecomp}
    p(a,b,\lambda|x,y) = p(b|a,\lambda,x,y)p(\lambda|a,x,y)p(a|x,y).
\end{equation}

The \textbf{No Retrocausality} assumption implies that we should have $p(a|x,y) = p(a|x)$ and $p(\lambda|a,x,y) = p(\lambda|a,x)$.  The first of these two conditions says that there is no signalling from future to past, which we expect the operational theory to obey in any case.  The second is the more substantive assumption, often called \emph{measurement independence}, and we expect this to be violated in a theory with retrocausality.  The no-retrocausality assumption implies that the probabilities can be decomposed as
\begin{equation}
    \label{eq:noretrodecomp}
    p(a,b,\lambda|x,y) = p(b|a,\lambda,x,y)p(\lambda|a,x)p(a|x).
\end{equation}

\begin{assumption}
    \textbf{$\bm{\lambda}$-mediation}: The ontic state $\lambda$ mediates any remaining correlation between the preparation and the measurement, i.e.\ $p(b|a,\lambda,x,y) = p(b|\lambda,y)$ \footnote{In fact, our main result can be derived under the weaker assumption that the ontic state $\lambda$  and the preparation input $x$ mediate any remaining correlation between the preparation output and the measurement output, i.e.\ $p(b|a,\lambda,x,y) = p(b|\lambda,x,y)$.  However, we see no good reason to allow a direct causal connection from $x$ to $b$ while disallowing a causal connection from $a$ to $b$}.
\end{assumption}

In the language of influence diagrams, the arrows from $a$ and $x$ to $b$ are eliminated, and we are left with the causal structure illustrated in \cref{fig:mediation}.

\begin{figure}[!htb]
    \centering
    \includegraphics[width=0.9\columnwidth]{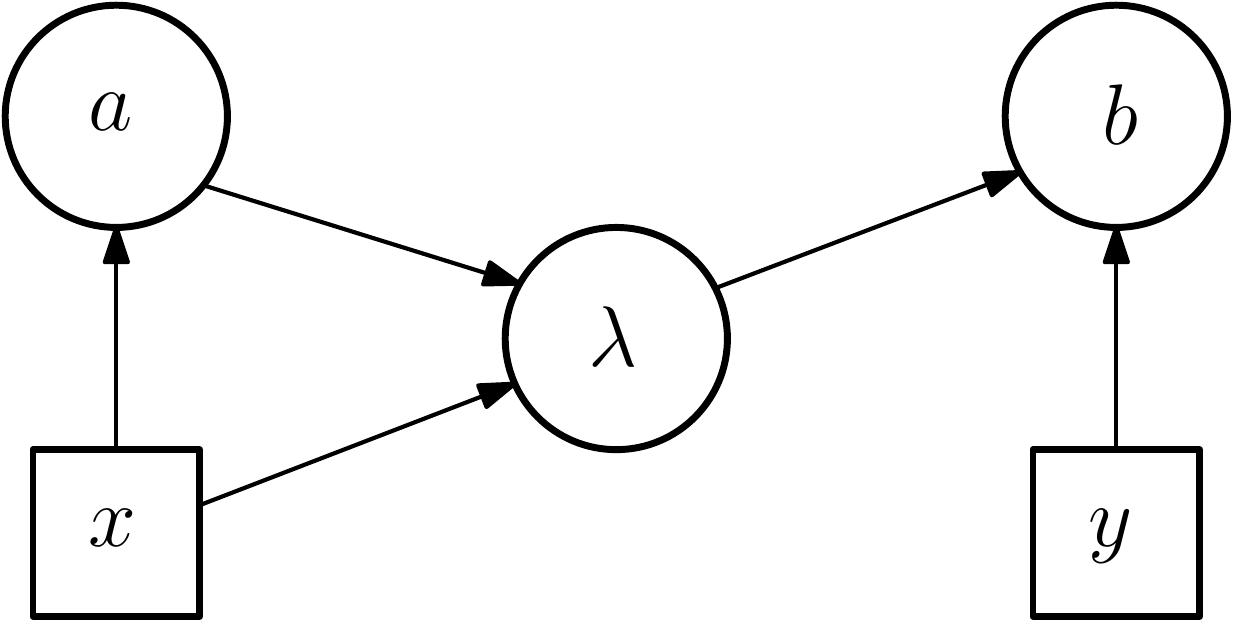}
    \caption{\label{fig:mediation}An influence diagram for an ontological model, which is an ontic extension that satisfies \textbf{No Retrocausality} and \textbf{$\bm{\lambda}$-mediation}.}
\end{figure}

The $\lambda$-mediation assumption is really just part of the definition of what we mean by an ontic state.  The properties of the system are supposed to be the cause of the correlation between preparation and measurement.  If there are other causal links between them, e.g.\ a telephone wire, then we are not really doing an experiment \emph{on} the system alone, but on the system in combination with something else.  If there are no such links, but we still allow correlation that is left unexplained by the ontic state, then we do not really have a realist model.

Together with \textbf{No Retrocausality}, this assumption implies that the probabilities can be decomposed as
\begin{equation}
    \label{eq:onticdecomp}
    p(a,b,\lambda|x,y) = p(b|\lambda,y)p(\lambda|a,x)p(a|x).
\end{equation}

\begin{definition}
    An ontic extension of an experiment that satisfies \textbf{No Retrocausality} and \textbf{$\bm{\lambda}$-mediation}, i.e.\ its probabilities decompose according to \cref{eq:onticdecomp}, is called an \emph{ontological model} of the experiment.
\end{definition}

For comparison with other work, note that $p(\lambda|a,x)$ is usually called the \emph{epistemic state} of the system and $p(b|\lambda,y)$ is called the \emph{response function} of the measuring device.  Normally, when proving theorems about realist accounts of quantum theory, the ontological models framework is simply assumed at the outset.  Here, we see that it is based on several nontrivial assumptions.

Note that, since Price's version of the argument does not assume \textbf{$\bm{\lambda}$-mediation}, it is of interest to discuss what we can prove without it.  We do so in \cref{app}

Note that we needed to state the \textbf{No Retrocausality} assumption before we could even formulate \textbf{$\bm{\lambda}$-mediation}.  In a theory with retrocausality, we would want to retain some notion that correlations between preparation and measurement are caused by the ontic properties of the system, but we have no license to infer conditional independencies from this if causal loops are possible.  Therefore, \textbf{$\bm{\lambda}$-mediation} would have to be replaced by something else.  What this might be will be discussed in future work.

\section{Ontological time symmetry}

\label{OnTS}

Our notion of ontological time symmetry can be formulated without assuming \textbf{No Retrocausality} or \textbf{$\bm{\lambda}$-mediation}.  This is fortunate, as it is something we might wish to preserve in a retrocausal theory.  Nonetheless, it is worth pausing to think about the role of $\lambda$ in an ontological model in order to understand how to formulate ontological time reversal.

For a classical system, $\lambda$ is often taken to be the phase space point occupied by the system, but this does not completely pin it down.  We could take $\lambda$ to be the phase space point immediately after the preparation, immediately before the measurement, or even let $\lambda(t)$ be the entire trajectory through phase space between preparation and measurement.  Given knowledge of the transformation $T$, which in this case is a specification of the Hamiltonian of the system, each of these choices is sufficient to screen off the preparation from the measurement.  There are many other choices besides these, such as specifying any part of the phase space trajectory, or dispensing with phase space entirely and specifying the configuration space point at two different times. Each of these choices would be equivalent and sufficient for proving our main result, as the screening off property is the main requirement.

Thus, in a general ontological model, we expect there to be multiple equivalent ways of specifying the ontic state.  However, they each have different implications for how to construct the time-reverse of an ontological model.  For example, if $\lambda$ is a specification of the properties of the system just after the preparation, then it should be thought of as a specification just before the measurement when we take the time reverse, and the direction of the momentum should be reversed.  If it is the entire trajectory then we also need to reverse the time direction of $\lambda(t)$ when we take the time-reverse.  In addition, it is conventional to allow other one-to-one transformations when we take a time reverse.  For example, the usual notion of time reversal in quantum theory involves mapping a quantum state $\ket{\psi}$ to its complex conjugate $\ket{\psi^*}$, so a similar transformation of the ontic state should happen in an ontic model in which the quantum state part of the ontology.  For these reasons, we should not expect that the ontic state $\lambda$ or even its state space $\Lambda$ are invariant under time reversal.  Instead, there should be a one-to-one mapping between the $\lambda$'s in the two different experiments.

Although an ontic extension need not be an ontological model, we still want to conceive of the ontic states as the same sorts of things in general.  Thus, we should allow a one-to-one correspondence between $\lambda$'s when we take the time reverse of an ontic extension.

\begin{definition}
    \label{def:ontTR}
    In an ontic extension of an operational theory, the ontic extension $(P,M,T,\Lambda)$ of an experiment $(P,M,T)$ has an \emph{ontological time reverse} if there exists another experiment $(P',M',T')$, with ontic extension $(P',M',T',\Lambda')$, where $P'$ has the same sets of inputs and outputs as $M$, $M'$ has the same sets of inputs and outputs as $P$, there exists a one-to-one map $f:\Lambda \rightarrow \Lambda'$, and
    \begin{equation}
        p_{P'M'T'}(b,a,f(\lambda)|y,x) = p_{PMT}(a,b,\lambda|x,y).
    \end{equation}

    An ontic extension of an operational theory is \emph{ontologically time symmetric} if every experiment has an ontological time reverse.
\end{definition}

Clearly, an ontic extension of an operational theory can only be ontologically time symmetric if the operational theory is itself operationally time symmetric, since we can obtain the operational predictions just by marginalizing over $\lambda$.

\begin{remark}
    \label{rem:samelambda}
    It is also true, mathematically at least, that if the ontic extension of an experiment has an ontological time reverse, then there also exists an extension in which $f$ is the identity.  We can construct such an extension from an arbitrary one by identifying the ontic state spaces of the time reverse pairs of experiments and setting $p_{P'M'T'}(b,a,\lambda|x,y)$ in the new extension to the value of $ p_{P'M'T'}(b,a,f(\lambda)|y,x)$ in the old extension.  Physically, these two extensions may tell very different stories about what happens between preparation and measurement, but mathematically we can always assume that $f$ is trivial without loss of generality, and will do so in what follows.
\end{remark}

\section{The time symmetry assumption}

\label{TSA}

We are now in a position to state our main assumption.

\begin{assumption}
    \textbf{Time Symmetry}: An ontic extension of an operational theory satisfies \textbf{Time Symmetry} if, whenever an experiment $(P,M,T)$ has an operational time reverse, the extension $(P,M,T,\Lambda)$ has an ontological time reverse.

    In particular, an ontic extension of an operationally time symmetric theory must be ontologically time symmetric.
\end{assumption}

The idea of this is that the most natural explanation for an operational time-symmetry in a theory is that it a reflection of an ontological time-symmetry.  Note that an experiment might have an operational time reverse for trivial reasons if it only records highly coarse-grained information about the true ontic state of affairs.  For example, suppose that the inputs of our preparation and measurement consist of choosing the position of a table in a room and the outputs consist of counting the number of plates on the table.  Suppose that all of the transformations we can do involve moving the table around violently so that some of the plates fall to the ground and are smashed.  This is a fundamentally time asymmetric theory (at the level of probabilities), but the asymmetry might not show up at the operational level if some of the transformations also involve bringing new plates from the kitchen and placing them on the table.  We would not expect the operational time symmetry of this theory to be reflected in the ontological description, as we know that there are plates being smashed but none being unsmashed.

The reason why this example is problematic is that the preparations and measurements that we can do are only giving us access to highly restricted information about the ontic state.  In particular, we are not being allowed to look at the floor to see if there are any smashed plates there.  The operational time symmetry is an accidental consequence of our restricted access, rather than a reflection of a fundamental symmetry.

However, if the no-signalling sector of an entire physical theory is operationally time symmetric then the time symmetry assumption seems much more well founded.  Recall that restricting to the no-signalling sector is merely designed to factor out the subjective arrow of time of the experimenter, which we do not believe to be a fundamental asymmetry, but is rather a consequence of the thermodynamic arrow.  If, after doing this, there is absolutely nothing you can do to determine the time direction, then this seems like grounds for positing a fundamental time symmetry.  It is only in this situation that we want to posit the time symmetry assumption.  \Cref{thm:QTS} shows that this is precisely the situation for quantum theory.

\section{Main results}

\label{MR}

In this section, we will be discussing a single experiment and its time reverse, which, due to Definitions~\ref{def:opTR} and \ref{def:ontTR}, and Remark~\ref{rem:samelambda}, can be taken to have the same ontic state space $\Lambda$ and the same probability distribution.  For this reason, we can drop the subscripts $P$, $M$, $T$ on the probability distribution.

\subsection{Main theorem}

\label{MT}

Our main theorem shows that an ontological model (i.e.\ an ontic extension satisfying \textbf{No Retrocausality} and \textbf{$\bm{\lambda}$-mediation}) of an experiment that satisfies \textbf{Time Symmetry} must obey a temporal analogue of Bell's local causality condition.

\begin{theorem}
    \label{thm:main}
    Let $(P,M,T)$ be an experiment that has an operational time reverse.  If its ontic extension satisfies \textbf{No Retrocausality}, \textbf{$\bm{\lambda}$-mediation}, and \textbf{Time Symmetry} then
    \begin{equation}
        \label{eq:lc}
        p(a,b,\lambda|x,y) =  p(a|x,\lambda)p(b|y,\lambda)p(\lambda).
    \end{equation}
\end{theorem}

To understand the role of \textbf{$\bm{\lambda}$-mediation}, it is helpful to first see what we can prove without it.

\begin{lemma}
    \label{lem:nolambda}
    Let $(P,M,T)$ be an experiment that has an operational time reverse.  If its ontic extension satisfies \textbf{No Retrocausality} and \textbf{Time Symmetry} then
    \begin{align}
        p(\lambda|x,y) & = p(\lambda), \label{eq:measind} \\
        p(b|\lambda,x,y) & = p(b|\lambda,y), \label{eq:param1} \\
        p(a|\lambda,x,y) & = p(a|\lambda,x). \label{eq:param2}
    \end{align}
\end{lemma}

\begin{proof}
    By \textbf{No Retrocausality}, the probabilities decompose as
    \begin{equation}
        \label{eq:decomp1}
        p(a,b,\lambda|x,y) = p(b|\lambda,x,a,y)p(\lambda|a,x)p(a|x).
    \end{equation}
    Using Bayes' rule, we have
    \begin{equation}
        p(\lambda|a,x) = \frac{p(a|\lambda,x)p(\lambda|x)}{p(a|x)}.
    \end{equation}
    Substituting this back into \cref{eq:decomp1} gives
    \begin{equation}
        \label{eq:Bayes1}
        p(a,b,\lambda|x,y) = p(b|\lambda,x,a,y)p(a|\lambda,x)p(\lambda|x).
    \end{equation}
    Summing over $a$ and $b$ then gives $p(\lambda|x,y) = p(\lambda|x)$.  By \textbf{Time Symmetry}, we also have the decomposition
    \begin{equation}
        \label{eq:decomp2}
        p(a,b,\lambda|x,y) = p(a|\lambda,x,y,b)p(\lambda|b,y)p(b|y),
    \end{equation}
    and applying the same argument to this gives $p(\lambda|x,y) = p(\lambda|y)$.  We therefore have $p(\lambda|x,y) = p(\lambda|x) = p(\lambda|y)$, but this means that it cannot depend on either $x$ or $y$, so $p(\lambda|x,y) = p(\lambda)$.

    Substituting this into \cref{eq:Bayes1} and summing over $b$ gives
    \begin{equation}
        p(a,\lambda|x,y) = p(a|\lambda,x)p(\lambda).
    \end{equation}
    Dividing both sides by $p(\lambda|x,y)$, which is equal to $p(\lambda)$, gives
    \begin{equation}
        p(a|\lambda,x,y) = p(a|\lambda,x).
    \end{equation}
    as required.  Applying the same argument to the time-reversed decomposition will then give
    \begin{equation}
        p(b|\lambda,x,y) = p(b|\lambda,y)
    \end{equation}
\end{proof}

\begin{proof}[Proof of \cref{thm:main}]
    In the context of a spacelike Bell experiment, \cref{eq:measind} is known as \emph{measurement independence}, and equations \eqref{eq:param1} and \eqref{eq:param2} are known as \emph{parameter independence}.  These conditions, together with
    \begin{align}
        \label{eq:outind}
        p(b|a,\lambda,x,y) = p(b|\lambda,x,y),
    \end{align}
    which is known as \emph{outcome independence}, are known to imply Bell's local causality condition \cite{Jarrett1984, Shimony1984, Shimony1993}.  However, in the timelike context, \cref{eq:outind} is an instance of \textbf{$\bm{\lambda}$-mediation}.  In fact, we have the stronger condition that $p(b|a,\lambda,x,y) = p(b|\lambda,y)$, so this combined with \cref{lem:nolambda} implies \cref{eq:lc}.

    For completeness, we repeat the proof that these conditions imply local causality here.

    An ontic extension that satisfies \textbf{No Retrocausality} has decomposition
    \begin{equation}
        p(a,b,\lambda|x,y) = p(b|a,\lambda,x,y)p(\lambda|a,x,y)p(a|x).
    \end{equation}
    By the Bayes' rule argument used in the proof of \cref{lem:nolambda}, we can rewrite this as
    \begin{equation}
        \label{eq:noretrobayesdecomp}
        p(a,b,\lambda|x,y) = p(b|a,\lambda,x,y)p(a|\lambda,x,y)p(\lambda|x).
    \end{equation}
    By measurement independence, we have $p(\lambda|x) = p(\lambda)$.  By parameter independence we have $p(a|\lambda,x,y) = p(a|\lambda,x)$.  By outcome independence we have $p(b|a,\lambda,x,y) = p(b|\lambda,x,y)$ and then by parameter independence we have $p(b|\lambda,x,y) = p(b|\lambda,y)$.  Substituting all this into \cref{eq:noretrobayesdecomp} gives
    \begin{equation}
        p(a,b,\lambda|x,y) = p(b|\lambda,y)p(a|\lambda,x)p(\lambda),
    \end{equation}
    as required.
\end{proof}

\subsection{Example of a quantum violation}

\label{QE}

\Cref{eq:lc} is incompatible with some experiments in the no-signalling sector of quantum theory, for the same sort of reason as it is incompatible with Bell inequality violating experiments in the spacelike case. In particular, in the case where all the inputs and outputs are two-valued binary variables, it is well known that \eqref{eq:lc} implies the CHSH inequality \cite{Clauser1969, Brunner2013}.
\begin{equation}
    \label{eq:CHSH}
  \frac14 \sum_{x,y} p(a \oplus b = xy|x,y) \leq \frac34,
\end{equation}
where $\oplus$ is addition modulo $2$.

In the timelike context, the CHSH inequality can be violated using a single qubit, by taking $\Hilb[A] = \Hilb[B] = \mathbb{C}^2$, defining the projectors
\begin{equation}
    \left [ \theta \right ] = (\cos \frac{\theta}{2} \ket{0} + \sin \frac{\theta}{2} \ket{1})(\cos \frac{\theta}{2} \bra{0} + \sin \frac{\theta}{2} \bra{1}),
\end{equation}
and using
\begin{align}
    \rho_{0A|0} & = \frac{1}{2} [0], & \rho_{1A|0} & = \frac{1}{2} [\pi], \\ \rho_{0A|1} & = \frac{1}{2} \left [ \frac{\pi}{2} \right ]& \rho_{1A|1} & = \frac{1}{2} \left [ - \frac{\pi}{2} \right ],
\end{align}
\begin{align}
    E_{0|0B} & = \left [ \frac{\pi}{4} \right ] & E_{1|0B} & = \left [ -\frac{3\pi}{4} \right ] , \\
    E_{0|1B} & = \left [ -\frac{\pi}{4} \right ] & E_{1|1B} & = \left [ \frac{3\pi}{4} \right ],
\end{align}
and taking $\chn_{B|A}$ to be the identity superoperator.

The preparation is no-signalling because the ensemble average state for both inputs is the maximally mixed state $\rho_{A|0} = \rho_{A|1} = I_A/2$.  This experiment gives a value for the LHS of \cref{eq:CHSH} of $(\sqrt{2} + 2)/4 \approx 0.854$, in violation of the inequality.

\subsection{General quantum violation}

\label{QG}

In general, we would like to determine exactly which experiments in the no-signalling sector of quantum theory can have an ontic extension satisfying \textbf{No Retrocausality}, \textbf{$\bm{\lambda}$-mediation}, and \textbf{Time Symmetry}, and which cannot.  Specifically, given an initial ensemble average state $\rho_A$, and a CPT map $\mathcal{E}_{B|A}$, when can all experiments consisting of preparing an ensemble decomposition of $\rho_A$, applying the transformation $\mathcal{E}_{B|A}$, and measuring a POVM on $\Hilb[B]$, have an ontic extension satisfying these assumptions?  It turns out that this question is equivalent to asking when a bipartite state $\rho_{AB}$ on the tensor product $\Hilb[A] \otimes \Hilb[B]$ admits a local hidden variable theory, which is a question that has been studied extensively in the literature \cite{Augusiak2014}.

To see this, we again make use of the formalism described in \cite{Leifer2013a}.  Given a channel $\mathcal{E}_{B|A}$, we can define a bipartite operator $\rho_{B|A}$ on $\Hilb[A] \otimes \Hilb[B]$, known as a \emph{conditional state}, via the Choi-Jamio{\l}kowski isomorphism as follows
\begin{equation}
    \rho_{B|A} = \mathcal{I}_A \otimes \mathcal{E}_{B|A'} \left ( \Ket{\Phi^+}\Bra{\Phi^+}_{AA'} \right ),
\end{equation}
where $A'$ is a system with the same Hilbert space as $A$,  $\Ket{\Phi^+}_{AA'} = \sum_j \Ket{jj}_{AA'}$, and $\mathcal{I}_A$ is the identity superoperator on $\Hilb[A]$.  This is an isomorphism because the action of the CPT map can be recovered via
\begin{equation}
    \label{eq:Choirev}
    \chn_{B|A}(M_A) = \Tr[A]{\rho_{B|A}^{T_A} M_A \otimes I_B},
\end{equation}
for any operator $M_A$ on $\Hilb[A]$, and where $^{T_A}$ is partial transpose in the $\ket{j}$ basis.

One can then define a bipartite state on $\Hilb[A] \otimes \Hilb[B]$ via
\begin{equation}
    \label{eq:bipartite}
    \rho_{AB} = \left ( \rho_A^{\frac{1}{2}} \otimes I_B \right ) \rho_{B|A} \left ( \rho_A^{\frac{1}{2}} \otimes I_B \right ).
\end{equation}
This state has $\rho_A = \Tr[B]{\rho_{AB}}$ as its reduced state, and the action of $\mathcal{E}_{B|A}$ on $\text{supp}(\rho_A)$ can be recovered by inverting \cref{eq:bipartite} and then applying \cref{eq:Choirev}.  We can thus consider this an isomorphism between pairs $(\rho_A,\mathcal{E}_{B|A})$ and bipartite states $\rho_{AB}$.  Although the action of  $\mathcal{E}_{B|A}$ outside $\text{supp}(\rho_A)$ is not recovered by the inverse transformation, it is irrelevant for determining the probabilities for experiments with ensemble average state $\rho_A$, so we will get a true isomorphism of the experimental probabilities.

Specifically, for any set of ensemble decompositions of $\rho_A$ into states $\rho_{aA|x}$ and any set of POVMs $E_{b|yB}$ on $\Hilb[B]$, it is straightforward to check that
\begin{multline}
    \Tr{E_{b|yB}\mathcal{E}_{B|A}(\rho_{aA|x})} \\ = \Tr{E_{a|xA} \otimes E_{b|yB} \rho_{AB}},
\end{multline}
where $E_{a|xA}$ is defined in terms of $\rho_A$ and $\rho_{aA|x}$ via \cref{eq:retroPOVM}.  We can conversely take any POVM on $E_{a|Ax}$ and invert \cref{eq:retroPOVM} to obtain ensemble decompositions.  The right hand side of this equation represents the predictions for local measurements made on the bipartite state $\rho_{AB}$.  Applied to this experiment, \cref{eq:lc} is just Bell's local causality in its original context of spacelike EPR-type experiments.  Thus, a state-CP map pair $(\rho_A,\mathcal{E}_{B|A})$ admits an ontic extension satisfying \textbf{No Retrocausality}, \textbf{$\bm{\lambda}$-mediation}, and  \textbf{Time Symmetry} exactly when the predictions for local measurements on $\rho_{AB}$ admit a local hidden variable theory.

\section{Price's argument}

\label{Price}

As given in \cite{Price2012}, Price's argument is based on three assumptions, which he calls \emph{Realism}, \emph{Time-symmetry}, and \emph{Discreteness}.  There are also unstated background assumptions, which are equivalent to our \textbf{Realism} and \textbf{Free Choice} assumptions.  Since these are common to both arguments, we will leave them implicit in what follows.

Price summarises his argument as
\begin{quote}
\emph{Realism} + \emph{Time-symmetry} + \emph{Discreteness} $\implies$ Retrocausality.
\end{quote}
To give this a similar logical form to our result, this could be rewritten as
\begin{quote}
\textbf{No Retrocausality} + \emph{Realism} + \emph{Time-symmetry} + \emph{Discreteness} $\implies$ Contradiction (with quantum theory)
\label{huwarg}
\end{quote}

As noted in the introduction, Price's \emph{Realism} assumption is an assumption of the reality of the quantum state.  In our terminology, an ontic extension of a quantum experiment is called \emph{$\psi$-ontic} if, in the case where the states $\rho_{aA|x}$ are pure, i.e.\ $\rho_{aA|x}= p(a|x)\ket{\psi_{a|x}}\bra{\psi_{a|x}}_A$, the distributions $p(\lambda|a,x)$ and $p(\lambda|a'x')$ have disjoint supports whenever $\proj{\psi_{a|x}}$ and $\proj{\psi_{a'|x'}}$ are distinct \footnote{Measure-theoretic technicalities of this definition are discussed in \cite{Leifer2014}}.  An extension that is not $\psi$-ontic is called \emph{$\psi$-epistemic}.  We call the assumption that an extension should be $\psi$-ontic, the \textbf{$\bm{\psi}$-ontology} assumption.

As discussed in the introduction, Price's \emph{Discreteness} assumption is needed to argue that a photon exiting a beamsplitter and then being detected on one of the output ports is the time reverse of a photon being inserted into a beamsplitter along a definite input port.  This guarantees that the operational time reversal used in Price's argument can be identified as a true time reverse.  His \emph{Time-symmetry} assumption is then that this time reverse should be an ontological time reverse as well.  In fact, Price's experiment is its own operational time reverse, i.e.\ $P'=P$, $M'=M$ and $T'=T$, so he only needs to apply the assumption in this case.  Note that, if we time reverse the density operators and POVMs from the experiment we described in \S\ref{QE} using \cref{eq:retroPOVM} and \cref{eq:retrostates}, we do not obtain the same experiment.  However, this experiment has more than one operational time reverse because we can apply a unitary $U$ to the states and its inverse $U^{\dagger}$ to the POVM elements without changing the operational predictions.  By doing this, we can rotate the states and measurements back to the original ones, and we see that the original experiment is indeed its own operational time reverse.  Thus, we could restrict the \textbf{Time Symmetry} assumption to experiments that are their own time reverse in our argument as well, although we see no good reason for doing so.  In any case, our \textbf{Time Symmetry} assumption plays an equivalent role to the conjunction of \emph{Discreteness} and \emph{Time symmetry} in Price's argument.

Replacing Price's \emph{Realism}, \emph{Discreteness}, and \emph{Time symmetry} assumptions with our \textbf{$\bm{\psi}$-ontology} and \textbf{Time Symmetry} assumptions, the logic of Price's argument becomes
\begin{quote}
\textbf{No Retrocausality} + \textbf{$\bm{\psi}$-ontology} + \textbf{Time Symmetry} $\implies$ Contradiction (with quantum theory).
\label{ourhuwarg}
\end{quote}
Note that Price's argument does not require $\bm{\lambda}$\textbf{-mediation}, but we consider this a more reasonable assumption to make than $\bm{\psi}$\textbf{-ontology}, since $\psi$-epistemic retrocausal theories are of interest.  In fairness to Price, he does point out that retrocausal $\psi$-epistemic models are possible, and gives an example.  However, as explained shortly, without $\bm{\psi}$\textbf{-ontology}, the experiment considered by Price also has a non-retrocausal $\psi$-epistemic model, so his argument is not sufficient to show that such models should be retrocausal.

We can reconstruct the logic of Price's argument in our framework, with an even simpler experiment than he uses.  Consider an experiment that has no measurement, i.e.\ formally the input and output variables of the measurement can take only one possible value, so we can drop $y$ and $b$ from our probability distributions, and  represent the ``measurement'' by the trivial single-outcome POVM $I_B$.  The preparation has a binary input and output.  All that really matters for the argument is that the $\rho_{aA|x}$'s are four distinct pure states, but for definiteness we will use a qubit and the following four quantum states
\begin{align}
    \rho_{0A|0} & = \frac{1}{2} \left [ 0 \right ] & \rho_{1A|0} & = \frac{1}{2}\left [ \pi \right ], \label{eq:Price11} \\
    \rho_{0A|1} & = \frac{1}{2} \left [ \frac{\pi}{2} \right ] & \rho_{1A|1} & = \frac{1}{2} \left [ -\frac{\pi}{2} \right ]. \label{eq:Price12}
\end{align}

By the $\bm{\psi}$\textbf{-ontology} assumption, the four probability distributions $p(\lambda|a,x)$ each have disjoint support.  This means that the two distributions $p(\lambda|x)$ also have disjoint support because
\begin{equation}
    p(\lambda|x) = \sum_a p(\lambda|a,x)p(a|x),
\end{equation}
so, for a fixed $x$ we have a convex combination of two distributions and hence $\text{supp}(p(\lambda|x)) = \cup_{a} \text{supp}(p(\lambda|a,x))$, which are disjoint for different values of $x$.  In particular, this means that $\lambda$ and $x$ are not independent, i.e.\ $p(\lambda|x) \neq p(\lambda)$.

Now consider the operational time reverse of this experiment.  By \textbf{No Retrocausality} its ontic extension must satisfy $p(\lambda|x) = p(\lambda)$ because $x$ is now the measurement input and $\lambda$ is in its causal past.  However, by \textbf{Time Symmetry}, these probabilities must be the same as in the original experiment, so we have a contradiction.

The intuition behind this argument is that if the quantum state is real in the original experiment, in the sense of disjoint supports for the $p(\lambda|a,x)$'s, then the $p(\lambda|a,x)$'s must still have disjoint support in the ontological time reverse.  This implies that the retrodictive states, evolving backwards from the measurement to the preparation, must also be real, which contradicts \textbf{No Retrocausality}.

If, following Price, we only want to apply the \textbf{Time Symmetry} assumption to experiments that are their own operational time reverse, then we can add a measurement to the end of the experiment, with POVM elements
\begin{align}
    E_{0|0B} & = \left [ 0 \right ] & E_{1|0B} & = \left [ \pi \right ] , \label{eq:Price21} \\
    E_{0|1B} & = \left [ \frac{\pi}{2} \right ] & E_{1|1B} & = \left [ -\frac{\pi}{2} \right ]. \label{eq:Price22}
\end{align}
Using \cref{eq:retroPOVM} and \cref{eq:retrostates}, it is easily checked that this is its own operational time reverse.  This is a temporal version of the EPR experiment and is the example that Price actually uses.  However, the reasoning for this experiment would be exactly the same as above.  One reason for preferring this version is that, without the \textbf{No Retrocausality} assumption, \textbf{Time Symmetry} and \textbf{$\bm{\psi}$-ontology} now imply that both the usual forwards evolving quantum state and the backwards evolving retrodictive state must be real in the \emph{same} experiment.  However, the contradiction does not depend on having this extra feature.

\begin{table}
    \begin{tabular}{l|l}
         Assumptions & Experiment \\
         \hline
         \textbf{No Retrocausality}  & Prepare one of four \\
         $\bm{\psi}$\textbf{-ontology} & distinct pure states \\
         \textbf{Time Symmetry} & \\
         \hline
         \textbf{No Retrocausality} & \\
         $\bm{\psi}$\textbf{-ontology} & Timelike Bohm-EPR \\ 
         \textbf{Self Time Symmetry} & \\
         \hline
         \textbf{No Retrocausality} & \\
         $\bm{\lambda}$\textbf{-mediation} & Timelike Bell-CHSH \\
         \textbf{(Self) Time Symmetry} & 
    \end{tabular}
    \caption{\label{tbl:Summary}Summary of the assumptions and experiments involved in the two versions of Price's argument and in our argument.  \textbf{Self Time Symmetry} refers to the \textbf{Time Symmetry} assumption restricted to experiments that are their own operational time reverse.}
\end{table}

To conclude this section, we note that, without the \textbf{$\psi$-ontology} assumption, the experiment just described has a very simple ontological model, i.e.\ an ontic extension satisfying \textbf{No Retrocausality} and \textbf{$\bm{\lambda}$-mediation}.  This is just the Spekkens' toy theory \cite{Spekkens2007}, regarded as an ontological model for qubits prepared and measured in the eigenstates of Pauli operators.

Specifically, suppose that the ontic state space $\Lambda$ consists of two bits, $\lambda_0$ and $\lambda_1$.  We set the epistemic states $p(\lambda_0,\lambda_1|a,x) = \frac{1}{2} \delta_{\lambda_x,a}$, where $\delta$ is the Kronecker $\delta$-function, and the response functions $p(b|\lambda_0,\lambda_1,y) = \delta_{\lambda_y,b}$.  In addition, $p(a|x) = \frac{1}{2}$ for this experiment.  This gives the joint probabilities
\begin{multline}
    p(a,b,\lambda_0,\lambda_1|x,y) \\ = p(b|\lambda_0,\lambda_1,y)p(\lambda_0,\lambda_1|a,x)p(a|x) \\ = \frac{1}{4}\delta_{\lambda_y,b} \delta_{\lambda_x,a},
\end{multline}
which reproduces the operational predictions upon marginalizing over $\lambda_0$ and $\lambda_1$.  In addition, this distribution is invariant under the exchange of $(a,x)$ with $(b,y)$, so the operational time reverse, which is the same experiment in this case, is modelled by the same distribution, so \textbf{Time
Symmetry} is satisfied.

This model is $\psi$-epistemic as, for example, the distributions $p(\lambda_0,\lambda_1|a=0,x=0)$ and $p(\lambda_0,\lambda_1|a=0,x=1)$ both assign probability $1/2$ to the ontic state $\lambda_0 = 0, \lambda_1 = 0$.  It does, however, satisfy \textbf{$\bm{\lambda}$-mediation} because it is an ontological model.  This shows that \textbf{$\bm{\psi}$-ontology} is essential for Price's argument, but not well justified for the example he uses.

\section{Time Symmetry in de Broglie-Bohm and Everett}

\label{EdBB}

In the introduction, we mentioned that de Broglie-Bohm (dBB) theory and the Everett/many-worlds interpretation both satisfy the conventional notion of time symmetry of dynamical laws, but they are not counter-examples to our argument because \textbf{Time Symmetry} is a different notion.

The case of dBB is the easiest to deal with because it satisfies all of our assumptions other than \textbf{Time Symmetry}.  Consider an experiment performed on a spin-$\frac{1}{2}$ particle, where the preparation and measurement are of the spin.  The preparation prepares pure states $\ket{\psi_{x|a}}$ and the measurement is in orthonormal bases $\ket{\phi_{b|y}}$ (one basis for each value of $y$).  Suppose that the position state $\ket{\Psi(\vec{r})}$ remains uncorrelated with the spin between the preparation and measurement so that $\ket{\Psi(\vec{r})}$ is the same for all choices of $x$ and $a$.  In dBB, the ontic state of this system consists of two pieces: the quantum state $\ket{\Psi(\vec{r})}\otimes\ket{\psi_{a|x}}$ and the position $\vec{R}$ of the particle, i.e. $\lambda = (\ket{\Psi(\vec{r})}\otimes\ket{\psi_{a|x}},\vec{R})$.  The quantum state evolves according to the Schr{\"o}dinger equation, and the position evolves according to the guidance equation, both of which are symmetric under the replacement of $t$ with $-t$ \footnote{We also have to replace $\ket{\Psi}$ with $\ket{\Psi^*}$, but our argument works for states that only have real amplitudes.}.  Since the position state is the same for all the experiments we are considering, and hence the position $\vec{R}$ is always distributed in the same way, the violation of \textbf{Time Symmetry} is entirely due to the spin state.

Since this theory is $\psi$-ontic, it is susceptible to the arguments given for the example in the previous section.  Namely, given that $p(\lambda|x)$ depends on $x$ (i.e.~the spin part of the quantum state is either $\ket{0}$ or $\ket{\pi}$ if $x=0$, and either $\Ket{\frac{\pi}{2}}$ or $\Ket{-\frac{\pi}{2}}$ if $x=1$), $\lambda$, which includes the spin state, is necessarily different for each choice of $x$.  \textbf{Time symmetry} then implies that $\lambda$ should depend on the measurement choice as well, but it does not as the quantum state only depends on the choice of preparation, and the distribution of positions is the same for all experiments.

As shown by Price's argument, a $\psi$-ontic theory needs to include a backwards evolving retrodictive state, in addition to the usual predictive state, in its ontology in order to satisfy \textbf{Time Symmetry}.  For example, the transactional interpretation \cite{Cramer1986, Kastner2012} and the two state vector formalism \cite{Aharonov1964, Aharonov2008}, interpreted realistically, both posit such states and satisfy \textbf{Time Symmetry}.

The reason why this does not conflict with the dynamical time symmetry of dBB is that operational time symmetry and dynamical time symmetry disagree on what is the time reverse of an experiment.  To see this, it is useful to look at the dynamical time reverse of a prepare and measure experiment.  In order to exhibit dynamical time reversal, it is necessary to treat the preparation and measurement devices quantum mechanically, as dynamical time symmetry only holds when we do not apply the measurement postulates of quantum theory.  In any case, this is how measurement is supposed to be handled in dBB, as the measurement postulates are understood as effective rules that apply when there is sufficient decoherence in the position basis.

To do this, we introduce four systems in addition to the spin-$\frac{1}{2}$ particle $S$.  The preparation device has an input, described by a physical device such as a switch, which we label $X$.  $X$ can be in one of several orthogonal states $\ket{x}_X$ corresponding to the different choices of preparation input.  The output is described by a device such as a pointer, which we label system $A$, with orthogonal states $\ket{a}_A$ corresponding to the output.  Importantly, both $X$ and $A$ are macroscopic systems with the different $\ket{x}_X$ states having almost no overlap in position and similarly for the $\ket{a}_A$ states.  The measurement device is also described by two macroscopic systems, $Y$ and $B$, describing its inputs and outputs, with macroscopically distinct orthogonal states $\ket{y}_Y$ and $\ket{b}_B$.  Prior to the experiment, the systems are in the state $\ket{x}_X \otimes \ket{0}_A \otimes \ket{y}_Y \otimes \ket{0}_B \otimes \ket{0}_S$, the important points being that the preparation and measurement inputs are set according to the choice of the experimenter, and the spin-$\frac{1}{2}$ particle and output systems are in fixed states.

The preparation can then be described by a unitary operation $U_{XAS}$, which only acts on $X$, $A$ and $S$ as
\begin{multline}
    U_{XAS} \ket{x}_X \otimes \ket{0}_A \otimes \ket{0}_S \\ = \ket{x}_X \otimes \sum_a \frac{1}{\sqrt{p(a|x)}} \ket{a}_A\otimes\ket{\psi_{a|x}},
\end{multline}
Similarly, the measurement can be described by a unitary operation $V_{YBS}$, which only acts on $Y$, $B$ and $S$ as
\begin{multline}
    V_{YBS} \ket{y}_Y \otimes \ket{0}_B \otimes \ket{\psi}_S \\ = \ket{y}_Y \otimes \sum_b \ket{b}_B \otimes \left [ \phi_{b|y} \right ] \ket{\psi}.
\end{multline}

$U_{XAS}$ maps states in which $A$ and $S$ are uncorrelated to an entangled superposition of different $S$ states correlated with the value of $A$, and similarly $V_{YBS}$ maps states in which $B$ and $S$ are uncorrelated to an entangled superposition of different $S$ states correlated with the value of $B$.  Their dynamical time reverses $U^{\dagger}_{XAS}$ and $V^{\dagger}_{YBS}$ do the opposite, mapping these correlated states to unentangled ones.  Thus, although preparations and measurements are operational time reverses of one another, they are not dynamical time reverses of one another, but rather both operations have the same sort of form.  The dynamical time reverse of a measurement would be a kind of preparation in which we are allowed to initially place the system in a superposition of different preparation outputs, rather than one in which $a$ takes a definite value.

The case of Everett/many-worlds is a little more subtle, as, unlike dBB, it does not satisfy the "single world" part of our (Single World) \textbf{Realism} assumption, so it does not fit into the framework we have used to prove our result.  Therefore, to make the argument that many-worlds does not satisfy \textbf{Time Symmetry} in rigorous detail, one would first have to reformulate \textbf{Time Symmetry} in a framework that allows experiments to have multiple co-existent outcomes at the ontological level.  Nonetheless, even without doing this, the intuitive idea of \textbf{Time Symmetry} is still clear: experiments that are operational time reverses of one another ought to be time reverses at the ontological level.  Because many-worlds posits the unitary evolution of the usual forwards-evolving quantum state as its entire ontology, it is $\psi$-ontic and does not satisfy \textbf{Time Symmetry} due to the lack of an ontic backwards evolving state.  However, this does not contradict dynamical time symmetry, which works the same way as in dBB. 

For most no-go theorems, such as Bell's theorem and results about contextuality, the move to a many-worlds framework completely evades the conclusions of the theorem.  It is therefore intriguing that this move does not help to evade our result.  

Readers who are happy with dynamical time symmetry may be inclined to reject our \textbf{Time Symmetry} assumption at this point.  However, the conceptual idea of time symmetry---that you cannot tell the difference between a video played forwards and played in reverse---is more fundamental than any particular mathematical formalisation of it.  In a theory in which the operational predictions are probabilistic, and in which there is no universal agreement about ontology, there will be several different formalisations of time symmetry.  Although dynamical time symmetry is the traditional notion, it does not have a unique claim to capture the basic concept.  Even if one does view dynamical time symmetry as more fundamental than operational time symmetry, the latter is still a symmetry that exists in quantum theory, and it lacks an explanation in theories that reject our \textbf{Time Symmetry} assumption.

\section{Relation to Spekkens Contextuality}

\label{Spekkens}

Spekkens' variant of noncontextuality states that experimental procedures that are operationally equivalent, in the sense of predicting the same probabilities, should be represented the same way in an ontological model \cite{Spekkens2005}.  In particular, two preparation procedures that always yield the same probabilities for every measurement should correspond to the same epistemic state, or probability distribution, over the ontic states.  This is known as preparation noncontextuality.

In the present context, when we consider the ensemble average over outputs of a no-signalling preparation, the operational probabilities obey $p(b|x,y) = p(b|x',y) = p(b|y)$ for all $x,x'$, which is just the no-signalling condition.  Therefore in a preparation noncontextual ontological model, we must have $p(\lambda|x) = p(\lambda|x') = p(\lambda)$.  In our main theorem, we derived this as a consequence of \textbf{Time Symmetry} instead.

In quantum theory, no-signalling preparations correspond to different ensemble decompositions of the same density operator $\rho_A$, so preparation noncontextuality says that, in an ontological model, $\rho_A$ should be represented by the same probability distribution over $\lambda$, regardless of which ensemble decomposition was used to prepare it.

Spekkens' showed that no ontological model can satisfy this instance of preparation noncontextuality for a particular setup involving decompositions of the maximally mixed state of a qubit into six pure states \cite{Spekkens2005}. It has also been shown that this instance preparation noncontextuality implies temporal Bell inequalities, such as the CHSH inequality we used here \cite{Spekkens2009}.  These results can therefore be used as alternative proofs of the impossibility of an ontological model satisfying \textbf{Time Symmetry}.

These preparation contextuality results have been demonstrated experimentally \cite{Spekkens2009, Mazurek2016}, so these experiments can also be viewed as demonstrations of our result.  However, we would like an experimental confirmation to be independent of the details of quantum theory, which raises some important issues.

Firstly, our time symmetry assumption is only supposed to be applied to experiments in the no-signalling sector of an operational theory.  Without assuming a background theory, like quantum theory, there is really no way of verifying that a preparation is no-signalling.  This is because, although every measurement we have so far managed to perform might reveal no signal from the preparation, there could always be some novel measurement, which we have not figured out how to do yet, that would yield a signal.  Tests of preparation contextuality face a similar issue, which is that two preparation procedures cannot be shown to be operationally equivalent by experiment alone.  If we are willing to make additional assumptions, such as the assumption that the measurements we can do are tomographically complete, then it is possible to devise robust tests of preparation contextuality that do not suffer from this problem \cite{Pusey2015}, and these tests could, in the same way, be used as robust tests of our main result.

However, we also think that the \textbf{Time Symmetry} assumption should only be applied to theories that have an operationally time symmetric no-signalling sector, in order to avoid the kind of accidental time symmetries discussed in \S\ref{TSA}.  This is truly impossible to verify experimentally, as it would involve testing every possible no-signalling experiment.  Therefore, whilst experimentally robust tests of some forms of preparation contextuality can also be viewed as tests of our result, we think that the main significance of our result is in the context of a specific theory, which may or may not be quantum, but is known to be operationally time symmetric on theoretical grounds.

\section{Discussion and Conclusions}

\label{Conc}

We have shown that there is no ontic extension of quantum theory that satisfies \textbf{No Retrocausality}, \textbf{$\bm{\lambda}$-mediation}, and our \textbf{Time Symmetry} assumption, i.e.\ the requirement that an operational time symmetry should imply an ontological one.  Requiring an ontic extension can be broken down into two assumptions: \textbf{Realism} and \textbf{Free Choice}.  Our main result shows that at least one of these five assumptions must be given up.

\textbf{Realism} and \textbf{Free Choice} are part and parcel of what is normally meant by a realist theory, and are shared by many no-go theorems, such as Bell's theorem.  Their status here is no different from in these other contexts.  Proponents of Copenhagenish interpretations deny that ontic states exist at all, so they would give up on \textbf{Realism}.  Since our intention is to investigate realist interpretations, we need not discuss this option further here.  However, one can potentially also give up \textbf{Realism} by denying that it is actually required for a realist interpretation.  For example, since our realism is single-world, the Everett/many-worlds interpretation does not satisfy it, but, as discussed in the introduction, Everett/many-worlds is not time symmetric in the sense we are interested in, so adopting it does not solve the problem.  Still, it is possible that a different interpretation that involves multiple worlds might exist that does have an appropriate time symmetry.  Another option is to adopt a non-standard logic or probability theory.  Whilst this idea is appealing, the interpretation of such theories is usually operational, and a compelling realist underpinning is currently lacking.

In our view, the idea that ontic states are responsible for correlations between preparation and measurement, which is the idea behind $\bm{\lambda}$\textbf{-mediation}, is also a core feature of a realist theory.  It encodes the idea that ontic states are supposed to explain what we see in experiments.  If we are contemplating giving up \textbf{No Retrocausality}, then it is not entirely clear how to formulate this assumption, as conditional independence of past and future variables given $\lambda$ is not a reasonable condition in a theory where $\lambda$ can be causally affected by things both in its past and in its future.  Nonetheless, some semblance of the idea that preparations and measurements are correlated because of $\lambda$ needs to be retained in a retrocausal theory, and it is a challenge for proponents of retrocausal views to formalize this in a meaningful way.

This leaves us with \textbf{Time Symmetry} and \textbf{No Retrocausality} as the most reasonable assumptions to give up.  Both of these assumptions can be viewed as special cases of the idea that an ontic extension should not be fine tuned, which makes it difficult to make a clear choice between them.

Loosely speaking, a no-fine tuning assumption means that, when the operational predictions have a specified property, e.g.\ a symmetry, we postulate that the underlying realist theory should also have the same property.  For example, the operational predictions of quantum theory do not allow us to send signals into the past, so we posit that there are no causal influences from future to past in the underlying theory.  This is the \textbf{No Retrocausality} assumption.  Similarly, the operational predictions of quantum theory obey a certain kind of time symmetry, so we posit this symmetry for the underlying theory as well.  This is our \textbf{Time Symmetry} assumption.

No fine tuning assumptions are extremely common in no-go theorems for quantum theory.  For example, Spekkens' noncontextuality assumption, i.e.\ procedures that are operationally identical should be identical in the underlying theory, is of this type, as is parameter independence in Bell's theorem, i.e.\ we cannot send signals superluminally, so our choice of measurement should not affect anything at the other wing of the experiment superluminally either \footnote{See \cite{Wood2015} for a detailed discussion of this in the case of Bell's theorem. The assumptions in Bell's theorem can also be motivated directly from physical reasoning, e.g. invoking special relativity, without any appeal to the operational predictions of quantum theory. One might likewise be able to provide physical motivations for both \textbf{Time Symmetry} and \textbf{No Retrocausality} that do not appeal to their operational counterparts. We set aside such motivations in our discussion and focus on the fine tuning narrative because we believe that narrative enables the clearest comparison of the two assumptions.}.

No fine tuning assumptions seem reasonable because, if a given property holds in the underlying theory, then we can accept that property as a fundamental physical principle.  For example, that there are no causal influences from future to past might be accepted as such a principle.  Then, the explanation for why the property holds at the operational level is just by appeal to that principle.  For example, if there are no causal influences from future to past at all, then obviously there must be no signalling from future to past.

However, when a property holds operationally, but not ontologically, then it cannot be because of a fundamental physical principle.  Instead, there is a tension that cries out for explanation.  For example, if there actually are influences from future to past then why is it that we cannot use those influences to send a signal?

This can only happen if the probabilities describing our ignorance of the true ontic state are fine tuned.  For example, if there are retrocausal influences in the theory, but no signalling into the past, then the correlations between the operational variables and the ontic states must be set in such a way that the possibility of signalling will be exactly washed out upon marginalizing over the ontic state.

Note that we expect our operational experiments to only reveal coarse-grained information about the ontic state, e.g.\ consider the classical example in which the ontic state is the microstate, but the operational degrees of freedom only probe the macrostate.  However, we do not expect our experiments to be so coarse grained that they completely miss key features of the underlying ontology.  For example, in statistical mechanics we can in fact detect many properties of the ontology by observing the small statistical fluctuations they cause, as in Brownian motion.  Therefore, if there are retrocausal influences, one would expect to be able to use them in some way to signal, even if this is not as straightforward as by directly observing the ontic state.  If not, there is an explanatory gap that needs filling.

Since both \textbf{No Retrocausality} and \textbf{Time Symmetry} are no-fine tuning assumptions, the case for each of them is equally compelling.  However, there are at least three possible ways of dealing with a fine tuning, which are appropriate in different circumstances, and can help to shed light on which one of them must go.

\begin{enumerate}
    \item We can accept them as brute facts.
    \item We can look for an alternative framework that is not fine-tuned.
    \item We can explain them as emergent.
\end{enumerate}

The first option is to just accept that the fine tuning is a prediction of our physical theory, and leave it at that.    If the ontological description is playing any explanatory role at all then it is going to differ from the operational description in at least some respects, some of which will be mere accidents.  Therefore, this might be an appropriate response if we are looking at some rather specific and convoluted property of the operational predictions that does not seem closely connected to fundamental physical principles.  However, both time symmetry and the lack of signalling into the past seem like general features of our physical world, which deserve a better explanation than this.

The second option is to say that the fine tuning does not actually exist in nature, but is rather an indication that one of our other assumptions is wrong.  Once we discard the, perhaps implicit, faulty assumption, we will be able to find a theory with no fine tunings.  In the present context, it is not completely out of the question to discard \textbf{Realism} so that we can retain some version of both \textbf{No Retrocausality} and \textbf{Time Symmetry}, but, as we have discussed, it is challenging to conceive of a realist framework that does not assume \textbf{Realism}.

This brings us to the third option, which is emergent fine tunings.  As an example of this, consider a universe that is in thermal equilibrium everywhere at a fixed temperature.  This universe contains many fine tunings, e.g.\ it is not possible to send a signal into the future in this universe, let alone the past.  The underlying statistical mechanical description says that there are possible non-equilibrium states in the theory that do not have this property, so the probabilities seem fine tuned in order to prevent signalling into the future.  However, in this case, we have a dynamical explanation of the fine tuning.  If the universe does start in a nonequilibrium state then the dynamics will, with high probability, eventually evolve the universe towards equilibrium.  This is, of course, a perfectly legitimate explanation of why we might see an apparent fine tuning.

If we accept this kind of explanation, then we also have to accept that the fine tuning is not based on a fundamental physical principle, and is likely to have been violated somewhere in the universe at some point, e.g.\ in the early universe shortly after the big bang, before the analogue of the equilibration process had time to take hold.  It is difficult to see how this could happen for \textbf{Time Symmetry}, which seems like a fundamental physical symmetry, but much easier for \textbf{No Retrocausality}.

Recall that the very definition of the subjective arrow of time is that we can remember the past, which seems fixed, but the future seems open and manipulable.  If signalling into both the past and the future were commonplace then there would be no definite subjective arrow of time, and arguably no possibility for subjective experience.  There are good reasons for believing that this arrow is necessarily aligned with the thermodynamic arrow of time, and we already believe that the thermodynamic arrow is emergent, to be explained by the initial conditions of the universe, rather than a fundamental principle of physics.  It is therefore plausible that the possibility of signalling into the past in a universe with retrocausality would be washed out by the same processes from which the thermodynamic arrow emerges.

In other words, given that we already need an asymmetry in the boundary conditions of the universe in order to account for the thermodynamic arrow of time, it seems plausible that this same asymmetry could be used to explain why there is signalling into the future, but no signalling into the past, in a universe with retrocausality.

We conclude that the most plausible response to our result, other than giving up \textbf{Realism}, is to posit that there might be retrocausality in nature.  At the very least, this is a concrete and little explored possibility that holds the promise of evading almost all no-go theorems in the foundations of quantum theory, so it should be investigated further.

\begin{acknowledgments}
We are grateful to Michael Dascal, Lev Vaidman, and the participants of ``Free Will and Retrocausality in a Quantum World'' (July 2014, Cambridge, UK) for helpful questions and discussions.  Research at Perimeter Institute is supported by the Government of Canada through the Department of Innovation, Science and Economic Development Canada and by the Province of Ontario through the Ministry of Research, Innovation and Science. ML acknowledges the support of the Foundational Questions Institute, FQXi ``Physics of What Happens'' grant number FQXi-RFP-1512 for the project ``Quantum Theory in the Block Universe''.
\end{acknowledgments}

\appendix

\section{Dropping $\lambda$-mediation}

\label{app}

If we drop \textbf{$\bm{\lambda}$-mediation}, then \cref{lem:nolambda} still applies, so the probabilities in an ontic extension of an experiment must satisfy
\begin{align}
    p(\lambda|x,y) & = p(\lambda), \label{aeq:measind} \\
    p(b|\lambda,x,y) & = p(b|\lambda, y), \label{aeq:param1} \\
    p(a|\lambda,x,y) & = p(a|\lambda,x). \label{aeq:param2}
\end{align}

In the context of a spacelike Bell experiment, \cref{aeq:measind} is called \emph{measurement independence}, which is an implication of \textbf{No Retrocausality} applied to that context.  \Cref{aeq:param1} and \cref{aeq:param2} are known as \emph{parameter independence} in the that context.  Thus, everything we can prove from measurement independence and parameter independence for a spacelike Bell experiment, we can also prove from \textbf{No Retrocausality} and \textbf{Time Symmetry} for a timelike experiment.

In particular, this means that we can apply the Colbeck-Renner theorem \cite{Colbeck2011, Colbeck2012a} (see \cite{Leifer2014} for a treatment in the ontological models framework).  From the assumptions of measurement independence and parameter independence, they prove that there is a particular Bell experiment for which the probabilities must obey
\begin{align}
    \label{eq:CR}
    p(a|x,\lambda) & = p(a|x), & p(b|y,\lambda) & = p(b|y).
\end{align}
This means that, although we cannot prove that an ontic extension is impossible for this experiment, knowing the ontic state does not enable you to make finer-grained predictions about the marginal probabilities for either side of the experiment than just knowing its input would.  However, they may still allow finer-grained predictions about how the two wings are correlated.  Thus, although we may still posit ontic states, they are not adding as much as we might have hoped to the theory. 

The experiment used by Colbeck and Renner can be converted into a timelike experiment using the construction given in \S\ref{QG}.  From this, we get the same conclusion in the timelike case as in the spacelike case: knowing $\lambda$ does not help to predict the preparation output if you know its input, and similarly for the measurement.  Thus, an ontic extension does not really add as much as we might have hoped to the predictive power of the theory.

The timelike experiment used to prove this result uses a qubit with the identity transformation between preparation and measurement.  It has inputs $x$ and $y$ that run from $0$ to $N-1$, and binary outputs.  The preparation is described by the states
\begin{equation}
    \rho_{aA|x} = \frac{1}{2} \left [ \left ( \frac{x}{N} + a \right ) \pi \right ],
\end{equation}
and the measurement by the operators
\begin{equation}
    E_{b|yB} = \left [ \left ( \frac{2y+1}{2N} + b \right ) \pi \right ].
\end{equation}
The Colbeck-Renner proof works in the limit that $N \rightarrow \infty$.  In this limit, the experiment involves preparing every pure state in the $X$-$Z$ plane of the Bloch sphere as we vary over $x$ and $a$, and we could easily modify the proof to work on any plane of the Bloch sphere.  Similarly, we could embed the proof in any two dimensional subspace of a higher dimensional Hilbert space.  Thus, $p(a|x, \lambda) = p(a|x)$ and $p(b|y,\lambda) = p(b|y)$ for an experiment that can prepare any pure state in quantum theory.

\bibliography{sym}
\end{document}